\documentclass[lettersize,journal]{IEEEtran}

\usepackage{amsmath,amsfonts}
\usepackage{array}
\usepackage[caption=false,font=footnotesize]{subfig} 

\usepackage{textcomp}
\usepackage{stfloats}
\usepackage{url}
\usepackage{verbatim}
\usepackage{cite}
\usepackage{graphicx}
\usepackage{float} 
\usepackage[hidelinks]{hyperref}
\usepackage{mathtools}
\usepackage{nicematrix}
\usepackage{balance}
\usepackage{enumitem}
\usepackage{tikz}
\usepackage{textcomp}
\usepackage{bm}
\usepackage[ruled, lined, longend, linesnumbered]{algorithm2e}
\usepackage{colortbl}

\usepackage{amsthm}

\newtheorem{prop}{Proposition}
\newtheorem{remark}{Remark}

\newtheorem{lemma}{Lemma}

\graphicspath{{../figures/}}

\hypersetup{ %
    breaklinks=true, %
    citecolor=black, %
    colorlinks=true, %
    linkcolor=black, %
    urlcolor=blue
 }

\begin{document}

\title{Modulating Retroreflector-Aided UAV-Based FSO/QKD Systems}

\author{
Duy N. Luong,
Duy-Tuan Dao,
Cuong T. Nguyen,~\IEEEmembership{Graduate Student Member,~IEEE,}
Hoang D. Le,~\IEEEmembership{Member,~IEEE,}
and Anh T. Pham,~\IEEEmembership{Senior Member,~IEEE} 
\thanks{\textcopyright 2026 IEEE. Personal use of this material is permitted. Permission from IEEE must be obtained for all other uses, in any current or future media, including reprinting/republishing this material for advertising or promotional purposes, creating new collective works, for resale or redistribution to servers or lists, or reuse of any copyrighted component of this work in other works. DOI: \href{https://doi.org/10.1109/LCOMM.2026.3722531}{10.1109/LCOMM.2026.3722531}}
\thanks{This research is funded by the Ministry of Education and Training under project number B2024.DNA.16}
\thanks{Duy N. Luong and Duy-Tuan Dao are with The University of Danang - University of Science and Technology, Vietnam (email: nhatduylqd@gmail.com; ddtuan@dut.udn.vn).}
\thanks{Cuong T. Nguyen, Hoang D. Le, and Anh T. Pham are with the University of Aizu, Aizuwakamatsu 965-8580, Japan (email: cuong98hp@gmail.com; hoangle@u-aizu.ac.jp; pham@u-aizu.ac.jp).}
\thanks{The corresponding author: Duy-Tuan Dao.}
}

\markboth{IEEE Communications Letters} 
{Shell \MakeLowercase{\textit{et al.}}: A Sample Article Using IEEEtran.cls for IEEE Journals}


\maketitle

\begin{abstract}
Unmanned aerial vehicles (UAVs)-based free-space optics (FSO)/quantum key distribution (QKD) systems require high-precision pointing mechanisms. This increases system complexity and limits rapid deployment for lightweight and energy-constrained UAVs. This paper proposes a modulating retroreflector (MRR)-equipped UAV architecture for BB84-QKD systems that enables simplified yet accurate tracking while relaxing pointing requirements. A realistic quantum channel model is developed, for which we newly derive the channel probability distribution of transmittance (PDT). Capitalizing on the derived channel PDT, several QKD performance metrics are analytically obtained. Numerical results verify the feasibility of the proposed MRR-aided UAV for practical QKD deployment, highlight its effectiveness over conventional UAV-ground systems, and validate the accuracy of the developed analytical framework.
\end{abstract}


\section{Introduction}
\label{sec:introduction}

Recent years have witnessed the rapid proliferation of unmanned aerial vehicles (UAVs) in critical applications, such as disaster/emergency response, defense, and surveillance \cite{UAV_survey_2026}. This creates an urgent demand for secure and resilient UAV communication systems. Quantum key distribution (QKD), which offers unconditional security based on the laws of quantum mechanics, has emerged as a promising solution to safeguarding UAV systems \cite{ismail2025finite, Dabiri_QKD_2025}. To this end, UAV-based free-space optics (FSO)/QKD systems have attracted significant research attention worldwide.  

A pressing concern for UAV-based FSO/QKD systems is severe pointing misalignment caused by UAV hovering fluctuations. This substantially degrades the achievable secret key rate (SKR) performance. Acquisition, tracking, and pointing (ATP) mechanisms can alleviate this critical issue at the cost of considerable complexity, weight, and power consumption. This poses significant challenges for compact and energy-constrained UAVs. A potential alternative solution is to replace the ATP unit on small UAVs with a modulating retro-reflector (MRR) \cite{why_MQW_2021}. Extensive studies have demonstrated the effectiveness of MRR-aided UAVs in alleviating pointing constraints in size- and energy-limited classical FSO systems \cite{trinh2021experimental}. 

Unlike classical MRR-aided FSO systems, integrating MRRs into QKD requires a fundamentally different system architecture to properly support QKD protocol operation. In particular, the retro-reflector must encode quantum states (e.g., BB84 polarization states) onto weak coherent pulses or single photons, rather than merely modulating the intensity of a classical interrogating beam. Although MRR is well established in classical UAV-based FSO systems \cite{Dabiri_TWC_2022, dabiri2025novel}, its integration into QKD remains in its infancy \cite{Vallone2015experimental, Rabinovich2018free, ref_8}. Notably, the authors in \cite{Vallone2015experimental} experimentally investigated the MRR-aided satellite QKD, in which both encoding and detection were performed at the ground station (GS), precluding a true MRR-QKD implementation. In \cite{Rabinovich2018free}, an MRR-QKD system using multiple quantum-well modulators was experimentally demonstrated over a simplified channel model. Motivated by \cite{Rabinovich2018free}, the authors in \cite{ref_8} proposed a low-complexity MRR-QKD implementation for the B92 protocol.

It is worth noting that existing studies on MRR-QKD have mainly focused on experimental demonstrations in satellite \cite{Vallone2015experimental} or terrestrial \cite{Rabinovich2018free, ref_8} systems under simplified channel assumptions, i.e., with constant channel losses. When MRR is integrated into UAV-based FSO/QKD, critical impairments, including hovering-induced beam-to-MRR misalignment, MRR orientation fluctuations, time-varying turbulence-induced fading, and angle-of-arrival (AoA) fluctuations, significantly degrade secret key performance. While an analytical framework for UAV-to-ground QKD systems has recently been reported in \cite{Dabiri_QKD_2025}, MRR-induced double-pass propagation results in distinct statistical characteristics of the end-to-end channel transmittance. To our best knowledge, a detailed analytical investigation of the end-to-end channel transmittance distribution and QKD performance of MRR-aided UAV-based FSO/QKD systems over realistic quantum channels remains limited.   

Motivated by this gap, this paper presents a dedicated analytical framework to model and accurately estimate the performance of QKD in MRR-aided UAV-based FSO/QKD systems. Notably, the main contributions are threefold. First, we propose a practical MRR-based QKD architecture for UAV-based FSO/QKD systems. Secondly, we develop a comprehensive analytical framework accounting for realistic quantum channel conditions. It allows the derivation of quantum bit error rate (QBER) and secret key rate (SKR) performance. Finally, we provide detailed insights into the impact of two-phase MRR channels on system performance.  

The remainder of this paper is organized as follows. Section~\ref{sec:system_model} presents system and channel models. In Section~\ref{sec:performance_analysis}, we analyze QKD performance. Section~\ref{sec:results} provides simulation results, followed by conclusions in Section~\ref{sec:conclusion}.

\section{System and Channel models}
\label{sec:system_model}

\subsection{System Descriptions}
\label{sec:system_desp}

\begin{figure}[t]
    \centering
    \includegraphics[width=\columnwidth]{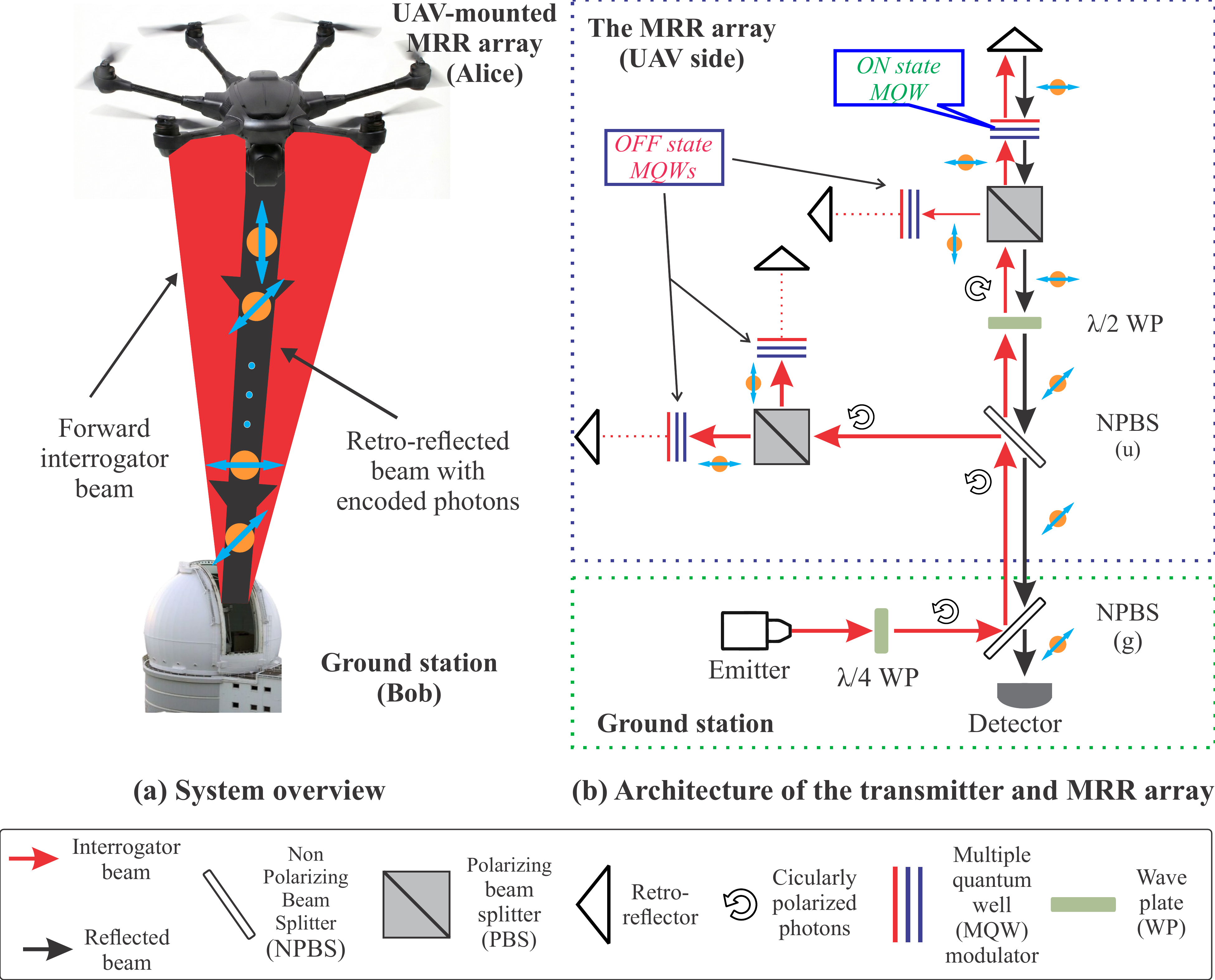}
    \caption{\textcolor{black}{The considered MRR-aided UAV-based QKD/FSO systems.}}
    \label{fig:System}
\end{figure}
The considered QKD system between a GS (Bob) and a UAV (Alice) is illustrated in Fig.~\ref{fig:System}(a). The UAV is equipped with an MRR, where Bob's laser interrogator serves as a photon source for the UAV. Particularly, Bob's laser interrogator emits a high-intensity optical pulse containing unencoded and circularly polarized photons toward Alice's MRR. Then, Alice's MRR encodes quantum states into photons and retro-reflects them. The MRR structure employed for the BB84 QKD protocol is depicted in Fig.~\ref{fig:System}(b). The MRR is configured as an array of four elements. Each contains a multiple quantum well (MQW) modulator\footnote{MRRs are well-suited for size, weight, and power (SWaP)-constrained \text{UAVs}, compared to liquid crystal (LC) or micro-electro-mechanical (MEMs) devices; MQW modulators offer significantly higher speeds \cite{why_MQW_2021}.} and a polarizer oriented to one of the BB84 polarization states (H, V, $\pm 45^\circ$).

The incoming circularly polarized interrogating beam enters the MRR array and is equally divided by a non-polarizing beam splitter (NPBS) into two optical paths, corresponding to two BB84 polarization bases, i.e., rectilinear (H-V) and diagonal ($\pm 45^\circ$) bases. Then, a polarizing beam splitter (PBS) in each path separates the circularly polarized beam into the corresponding linear polarization states. Here, for the diagonal basis path, a half-wave plate (HWP) oriented at $22.5^\circ$ precedes the PBS, rotating the retro-reflected H/V polarizations into $\pm 45^\circ$ states \cite{Rabinovich2018free}. Finally, Alice randomly\footnote{A quantum random number generator (RNG) is implemented to avoid RNG manipulation attacks.} selects one of four MRR elements to encode the quantum state, while the remaining three MRR elements do not retro-reflect the signal. Figure~\ref{fig:System}(b) shows an example, where the MRR corresponding to the diagonal $+45^\circ$ basis is selected and retro-reflects the encoded quantum state back to Bob.

To cope with the photon-number-splitting attack, we consider the two-decoy-state version of the BB84 protocol \cite{trinh2025optical}. Therein, Alice adjusts the MQW double-pass loss to create decoy pulses with the desired mean photon number \cite{Rabinovich2018free}. Moreover, we assume that Alice actively monitors the incident angle of the interrogating beam to avoid Trojan-horse and off-axis attacks \cite{ref_8}.
Also, the properties of the four MRRs are identical to avoid information leakage from photons' other degrees of freedom. Finally, Alice can utilize the photodetection capability of MQWs to determine whether the interrogating beam is controlled by Eve \cite{Rabinovich2018free}.

\subsection{Channel Model}

The end-to-end channel transmittance is modeled as $\tau_\text{total} = \eta \tau_\text{MRR} \tau_\text{l} I_\text{a} \tau_\text{p} \tau_\text{FoV}$, where $\eta$ is the Bob's detector efficiency, $\tau_\text{MRR}$ is the MRR reflectance, $\tau_\text{l}$ is the deterministic loss over the atmosphere, $I_\text{a}$ is the random intensity fluctuation caused by atmospheric turbulence, $\tau_\text{p}$ is the geometric and misalignment loss (GML), and $\tau_\text{FoV}$ is the AoA fluctuation. We assume that these channel coefficients are independent of each other\footnote{Since AoA fluctuation is mainly from the receiver-side FoV coupling and stabilization process, its correlation with UAV angular fluctuations is weak. Thus, $\tau_\text{FoV}$ is assumed to be independent of the UAV angular fluctuation \cite{Dabiri_QKD_2025}.}. Further details of them are as follows.

\subsubsection{Photon Over Atmospheric Channels}
The photon loss for each path is due to the atmospheric absorption and scattering, which can be modeled by the Beer-Lambert law as $\tau_{\text{l},i} = \exp(-\beta_i Z)$ \cite{trinh2021experimental}, where $Z$ is the link distance, {\color{black}$\beta_i = \frac{ \beta_{{\rm dB}, i} }{10^4 \log_{10} e}$} is the atmospheric extinction coefficient, and $i \in \{1, 2\}$ denote ground-to-UAV and UAV-to-ground links, respectively. The atmospheric loss over GS-UAV-GS paths is computed as $\tau_\text{l} = \tau_{\text{l},1} \tau_{\text{l},2}$. In addition, atmospheric turbulence causes random fluctuations in intensity. Here, $I_\text{a} = I_{\text{a},1} I_{\text{a},2}$. {\color{black} Let $\sigma_{\text{R},i}^2$ be the Rytov variance. In this work, we assume the weak turbulence condition, which is valid when $\sigma_{\text{R},i}^2 < 1$. } Thus, $I_{\text{a},i}$ is modeled by the log-normal distribution, i.e., \cite{trinh2021experimental}
\begin{align}
    \label{tur_1}
    f_{I_{\text{a},i}} (I_{\text{a},i}) = \frac{1}{I_{\text{a},i} \sqrt{2 \pi \sigma_{\text{R},i}^2}} \exp \left[ - \frac{\left(\ln{I_{\text{a},i}} + \frac{\sigma_{\text{R},i}^2}{2} \right)^2}{2 \sigma_{\text{R},i}^2} \right],
\end{align}
{\color{black}where $\sigma_{\text{R},1}^2$ and $\sigma_{\text{R},2}^2$ can be calculated as a function of GS height $Z_\mathrm{{hg}}$, elevation angle $\theta_{\mathrm{elv}}$, wind speed $V$, and refractive index structure $C_\mathrm{n}^2 \left( 0 \right)$ as in \cite[(14)]{ata2022performance} and \cite[(12)]{ata2022performance}, respectively. Moreover, $I_{\text{a},1}$ and $I_{\text{a},2}$ are two correlated random variables with a correlation coefficient $\rho$.}

\subsubsection{Effects of MRR on QKD Systems}
\label{sec:MRR_coefficient}
In this study, a passive corner-cube retro-reflector (CCR) is adopted for the MRR array. Each CCR consists of three perpendicular triangular mirrors. This leads to the photon loss, denoted as $\eta_\theta$, which is modeled by\cite[(33)]{Dabiri_TWC_2022}     
\begin{align}
    \label{eta_mrr1}
    f_{\eta_\theta}(\eta_\theta) \!\! \simeq \!\! \frac{1}{\eta_\theta \sqrt{2 \pi \ln \!\! {\left(1 \!\! + \!\! \frac{\sigma_\theta^2}{\mu_\theta^2}\right) } } } \!\! \exp{ \left[ - \frac{ \ln^2 \!\! { \left( \frac{\eta_\theta \sqrt{\mu_\theta^2 + \sigma_\theta^2}}{\mu_\theta^2} \right)} }{2 \ln{\left(1 + \frac{\sigma_\theta^2}{\mu_\theta^2}\right)} } \right]},
\end{align}
where $\mu_\theta$ and $\sigma_\theta^2$ are respectively the mean and variance of $\eta_\theta$, which can be derived as a function of the standard deviation (SD) of the UAV's angular fluctuation \cite[Table~II]{Dabiri_TWC_2022}. Other components also introduce internal photon loss, denoted as $\tau_\mathrm{int} = L_{\mathrm{NPBS}, \text{u}}^2 L_\text{PBS} L_\text{MQW}$. Here, $L_{\mathrm{NPBS}, \text{u}}$, $L_\text{PBS}$, and $L_\text{MQW}$ present the NPBS splitting factors at Alice (UAV), the PBS splitting ratio, and the MQW double-pass loss, respectively. Then, the MRR reflectance is given as $\tau_\text{MRR} = \eta_\theta \tau_\mathrm{int}$.

\subsubsection{Geometric and Misalignment Loss} 

The GML coefficient is given as $\tau_\text{p} = \tau_{\text{p},1} \tau_{\text{p},2}$, where $\tau_{\text{p},1}$ and $\tau_{\text{p},2}$ represent for forward and retro-reflected links, respectively. In the forward link, imperfect tracking leads to the pointing error with an SD of $\sigma_{\theta_e}$. The radial displacement between the beam center and the MRR aperture center is then given as $d_\text{p} = \sqrt{d_{\text{p},x}^2 + d_{\text{p},y}^2}$, where $d_{\text{p},x} \simeq Z \theta_{\text{e},x} $ and $d_{\text{p},y} \simeq Z \theta_{\text{e},y} $ with $\theta_{\text{e},x} \sim \mathcal{N}(0, \sigma_{\theta_e}^2)$ and $\theta_{\text{e},y} \sim \mathcal{N}(0, \sigma_{\theta_e}^2)$. Assuming the Gaussian beam profile and that the interrogating beam footprint is much larger than the MRR aperture area $A_\text{r}$, the probability that a photon falls within the MRR aperture can be approximated as\cite[(5)]{Dabiri_TWC_2022}
\begin{align}
    \label{eqn:tau_1}
    \tau_{\text{p},1} \simeq \frac{2A_\text{r}}{\pi w_{z,\text{1}}^2} \exp \left(-\frac{2d_\text{p}^2}{w_{z,\text{1}}^2} \right),
\end{align}
where $w_{z,\text{1}} = w_0 \sqrt{1+ \left(\frac{\lambda Z}{\pi w_0^2}\right)^2}$ is the interrogating beam waist at the distance $Z$, $w_0$ is the beam waist at $Z = 0$, and $\lambda$ is the optical wavelength. The distribution of $d_\text{p}$ is well approximated by Rayleigh, from which the PDT of $\tau_{\text{p},1}$ can be obtained as\cite[(37)]{Dabiri_TWC_2022} 
\begin{align}
f_{\tau_{\text{p},1}} (\tau_{\text{p},1}) = K \left(\frac{\pi w_{z,\text{1}}^2}{2 A_\text{r}}\right)^K \tau_{\text{p},1}^{K-1}, 
\end{align}
where $K = \frac{w_{z,\text{1}}^2}{4 Z^2 \sigma_{\theta_e}^2}$, and $\tau_{\text{p},1} \in \left(0, \frac{2A_\text{r}}{\pi w_{z,\text{1}}^2}\right]$.
In the retro-reflected link, it is seen that MRR reflects the laser signal at the same input angle. Thus, the center of the retro-reflected beam is located approximately at that of the receiver’s aperture \cite{Dabiri_TWC_2022}. Using the grid-based discretization scheme, the photon capture probability at Bob's detector is given as\cite[(17)]{Dabiri_QKD_2025}
\begin{align}
    \label{grid_p_2}
    \tau_{\text{p},2} \!\! \simeq \!\! \sum_{j=1}^{N_\text{g}} \frac{2 \Delta x}{\sqrt{2 \pi} w_{z,\text{2}}} \text{erf} \!\! \left(\sqrt{\frac{2 (r_\text{g}^2 - x_j^2)}{w_{z,\text{2}}^2}} \right) \exp \!\! \left(-\frac{2x_j^2}{w_{z,\text{2}}^2} \right),
\end{align}
where the receiver's aperture interval $[-r_\text{g}, r_\text{g}]$ is equally divided into $N_\text{g}$ segments with the size of $\Delta x = 2r_\text{g}/N_\text{g}$, and $x_j$ is the center of the $j$-th segment. Also, $w_{z,\text{2}}$ is the retro-reflected beam waist, in which $w_{z,\text{2}} \simeq \xi_\text{div,2} Z$ with $\xi_\text{div,2}$ the MRR's divergence half-angle found in \cite[(10)]{Liu_MRR_AO_2017}.

\subsubsection{Aperture Coupling Efficiency}
Although the receiver aperture collects photons, only those within the quantum detector's field-of-view (FoV) are decoded. Here, $\tau_\text{FoV}$ corresponds to the probability that a single photon is successfully detected by a single-photon avalanche diode (SPAD) within its FoV. The PDT of $\tau_\text{FoV}$ is expressed as\cite[(35)]{Dabiri_QKD_2025}
\begin{align}
    f_{\tau_\text{FoV}} (\tau_\text{FoV}) = P_\mathrm0 \, \delta(\tau_\text{FoV}) + P_\mathrm1 \delta \left( \tau_\text{FoV} - 1 \right), 
\end{align}
where $P_\mathrm0 = \exp\!\left(-\dfrac{\theta_{\mathrm{FoV}}^{2}}{2\sigma_{\mathrm{AoA}}^{2}}\right)$, $P_1=1-P_0$, $\delta(\cdot)$ denotes the Dirac delta function, $\theta_{\mathrm{FoV}}$ and $\sigma_{\mathrm{AoA}}^2$ are respectively the detector's half-angle FoV and variance of angular misalignment.
\section{Performance Analysis}
\label{sec:performance_analysis}

\subsection{Photon Retro-Reflected Rate}
\label{sec:phton_rr_rate}

The average number of retro-reflected photons per pulse by the MRR array is calculated as
\begin{align}
    \mu = \mu_\text{emit} {\color{black}L_{\mathrm{g,NPBS_{R}}}} \tau_{\text{l},1} \mathbb{E}[\tau_\text{MRR}] \mathbb{E}[I_{\text{a},1}] \mathbb{E}[\tau_{\text{p},1}],  
\end{align}
where {\color{black}$L_{\mathrm{g,NPBS_{R}}}$} is the reflectance ratio of the GS's NPBS, $\mathbb{E}[\cdot]$ is the mean value, and $\mu_\text{emit}$ is the mean emitted photon number per pulse. It is computed as $\mu_\text{emit} = \frac{P_\text{emit} \lambda}{f_\text{rep} h c}$, where $P_\text{emit}$ is the emitted power, $f_\text{rep}$ is the pulse repetition rate, $h$ is the Planck’s constant, and $c$ is the light velocity. Here, $\mathbb{E}[\tau_\text{MRR}] = \mu_\theta \tau_\text{int}$, $\mathbb{E}[I_{\text{a},1}]$ = 1, and $\mathbb{E}[\tau_{\text{p},1}] = \frac{K}{K+1} \frac{2A_\text{r}}{\pi w_{z,\text{1}}^2}$.

\begin{remark}
To facilitate the performance analysis, we model Alice's MRR array as an effective weak coherent photon source to Bob's detector. The corresponding representative average photon number generation and the effective channel transmittance are denoted by $\mu_\text{n}$ and $\tau$, respectively.
\end{remark}

Here, the composite forward channel transmittance is given as $\tau_1 = \tau_{1,\text{det}} \tau_{1,\text{rand}}$, where 
{\color{black}
$\tau_{1,\text{det}} = L_{\mathrm{g,NPBS_{R}}} \tau_\mathrm{l,1} \tau_\mathrm{int}$, and $\tau_{1,\text{rand}} = I_\mathrm{a,1} \eta_\mathrm{\theta} \tau_\mathrm{p,1}$.
Then, $\mu_\text{n} = \mu_\text{emit} \tau_{1,\text{det}} \mathbb{E}[\tau_{1,\text{rand}}]$.
Noting that by changing $\tau_\mathrm{int}$ via adjusting $L_\text{MQW}$, Alice can obtain the desired mean photon number for signal and decoy pulses. 
Finally, $\tau$ is given as $\tau = \eta \frac{\tau_{1,\text{rand}}}{\mathbb{E}[\tau_{1,\text{rand}}]} I_{\text{a},2} \tau_{\text{p},2} \tau_{\text{l},2} L_{\mathrm{g,NPBS_{T}}} \tau_\text{FoV}$, where $L_{\mathrm{g,NPBS_{T}}}$ is the transmittance ratio of the GS's NPBS.
}

\subsection{Quantum Bit-error Rate (QBER)}
\label{sec:qber}

The average QBER over the effective channel transmittance $\tau$ is expressed by \cite[(14)]{Yudai} 
\begin{align}
\label{QBER_exact}
    \langle E_{\mu_\text{n}} \rangle =
    \frac{\displaystyle\int_{0}^{\infty} E_{\mu_\text{n}} Q_{\mu_\text{n}}(\tau) f_{\tau}(\tau) \mathrm{d} \tau}
         {\displaystyle\int_{0}^{\infty} Q_{\mu_\text{n}}(\tau) f_{\tau}(\tau) \mathrm{d} \tau},
\end{align}
where $E_{\mu_\text{n}} Q_{\mu_\text{n}}(\tau)$ and $Q_{\mu_\text{n}}(\tau)$ are found in \cite[(15)]{Yudai} and \cite[(16)]{Yudai}, respectively. To calculate \eqref{QBER_exact}, we need to find $f_{\tau}(\tau)$. 

\begin{lemma}
    The closed-form of the PDT for the effective channel transmittance $\tau$ is computed as  
    \begin{align}
        \label{eqn:eff_channel}
        f_{\tau}(\tau) = P_{0} \delta(\tau) + P_1 C_4 \tau^{K-1}
        Q\!\left(\frac{\ln (\tau) + C_5}{\sqrt{C_1}}\right),
    \end{align}
    where $Q(\cdot)$ is the tail distribution function of the Gaussian distribution. Additionally, $s_\mathrm{m}^2 = \ln \left( 1 + \frac{\sigma_\theta^2}{\mu_\theta^2} \right)$, ${\color{black}C_1 = s_\mathrm{m}^2 + \sigma_\mathrm{R,1}^2 + \sigma_\mathrm{R,2}^2 + 2 \rho \sigma_\mathrm{R,1} \sigma_\mathrm{R,2}}$, ${\color{black}C_2 = 0.5( s_\mathrm{m}^2 + \sigma_\mathrm{R,1}^2 + \sigma_\mathrm{R,2}^2)}$, $C_3 = \kappa \eta \tau_\mathrm{p,2} \tau_\mathrm{l,2} {\color{black}L_{\mathrm{g,NPBS_{T}}}}$ with $\kappa = (K+1)/K$, $C_4 = K\,C_3^{-K} \exp \left(K^2 C_1/2 + K C_2 \right)$, and $C_5 = -\ln C_3 + K C_1 + C_2$.      
\end{lemma}

\begin{proof}
See Appendix~\ref{Appenx_A}.
\end{proof}

\begin{prop}
    Given $f_{\tau}(\tau)$, the approximation of QBER can be calculated as   
    \begin{align}
    \label{e_mu_m_cl}
    \langle E_{\mu_\text{n}} \rangle
    \!\! = \!\! \frac{\langle E_{\mu_\text{n}} Q_{\mu_\text{n}} \rangle}{\langle Q_{\mu_\text{n}} \rangle}
    \!\! = \!\! \frac{e_0 p'_\mathrm{dark}
    \!\!+ \!\! \big(e_\mathrm{pol} \!\! + \!\! e_0 P_\mathrm{AP}\big)\left[ 1 \!\! - \!\! G(\mu_\text{n}) \right]}
    {p'_\mathrm{dark}  + (1 + P_\mathrm{AP})\left[ 1 - G(\mu_\text{n}) \right]}.
    \end{align}
    where $p_\mathrm{dark}' = p_\mathrm{dark}(1+P_\mathrm{AP})$, $p_\mathrm{dark}$ is the dark count rate, $P_\mathrm{AP}$ is the after pulsing probability caused by the ultra-fast rate, $e_0$ is the error probability caused by noise photon detection events, and $e_\mathrm{pol}$ is the error rate due to polarization error. Additionally, the function $G \left( z \right)$ is expressed as follows
    \begin{align}
    G \left( z \right) \! = \! P_\mathrm0 \! + \! P_\mathrm1 \frac{C_4}{\sqrt{\pi}\,z^{K}} \sum_{i=1}^{n} w_i\,\gamma\!\left(K,\, \!\! z e^{\sqrt{C_1}\sqrt{2}\,x_i - C_5} \right),
    \label{G_approx}
    \end{align}  
    where $\gamma(\cdot,\cdot)$ is the lower incomplete gamma function and $\{x_i,w_i\}_{i=1}^{n}$ are the Gauss--Hermite nodes and weights.
\end{prop}

\begin{proof}
See Appendix~\ref{Appenx_B}.
\end{proof}

\begin{remark}
    The approximation used in \eqref{G_approx} quickly converges to its exact form expression with $n = 50$ terms.
\end{remark}

\subsection{Average Secret Key Rate (SKR)}
Using the decoy-state BB84 protocol, the information-theoretic lower limit for {\color{black}SKR\footnote{\color{black}This study focuses on performance-oriented analysis, where the SKR is applied under the assumption that the global optical phase of each retro-reflected pulse is random. Practical architecture to cope with random phase attacks will be considered in the future work.}} is given as \cite{trinh2025optical}
\begin{align}
    \label{eqn:skr}
    \mathcal{S}_{\mu_\text{n}} \!\! \ge \!\! \mathcal{R} s p d \Big\{ \!\!
    & - \!\! \langle Q_{\mu_\text{n}} \rangle f_\text{e} H_2(\langle E_{\mu_\text{n}} \rangle) \!\! + \!\! \langle Q_1^{L} \rangle \!\!
    \left[ 1 \!\! - \!\! H_2 \!\! \left(\langle e_1^{U} \rangle\right) \right] \!\! \Big\},
\end{align}
where $\mathcal{R}$, $s$, $p$, $d$, $f_\text{e}$, and $H_2(\cdot)$ denote repetition rate, sifting coefficient, parameter estimation coefficient, fraction of signal pulses, reconciliation efficiency, and Shannon's entropy function, respectively. Also, $\langle Q_1^{L} \rangle = \int_{0}^{\infty} Q_1^{L}(\tau) f(\tau) \mathrm{d} \tau$, where $Q_1^{L}(\tau)$ is in \cite[(G10)]{vasylyev2019satellite}.

\begin{figure*}[!t]
\normalsize
    \begin{align}
    \label{Q_l}
    \begin{split}
    \left\langle Q_1^{L} \right\rangle =
    \frac{\mu_\text{n}^{2} e^{-\mu_\text{n}}}{\mu_\text{n} \nu - \nu^{2}}
    \Bigg[
    \Big( p_{\mathrm{dark}}' + (1+P_{\mathrm{AP}})(1 - G(\nu)) \Big)e^{\nu} 
    -
    \Big( p_{\mathrm{dark}}' + (1+P_{\mathrm{AP}})(1 - G(\mu_\text{n})) \Big)
    e^{\mu_\text{n}}\frac{\nu^{2}}{\mu_\text{n}^{2}} 
    -
    \frac{\mu_\text{n}^{2} - \nu^{2}}{\mu_\text{n}^{2}}\, p_{\mathrm{dark}}
    \Bigg],
    \end{split}
    \end{align}
\hrulefill
\end{figure*}

\begin{lemma}
    The numerical expression of $\langle Q_1^{L} \rangle$ can be found in \eqref{Q_l}, where $\nu$ is the mean photon number of the decoy pulse.
\end{lemma}

To complete \eqref{eqn:skr}, we need to calculate $\langle e_1^{U} \rangle$, which is written as $\langle e_1^{U} \rangle = \frac{1}{\langle Q_1^{L} \rangle} \int_{0}^{\infty} e_1^{U}(\tau) Q_1^{L}(\tau) f(\tau) \mathrm{d}\tau$, where $e_1^{U}(\tau)$ is found in \cite{vasylyev2019satellite}. After several mathematical manipulations, $\langle e_1^{U} \rangle$ is computed as
\begin{equation}
\begin{split}
    \left\langle e_1^{U} \right\rangle
    &\!\! = \!\!
    \frac{
    \left[e_0 p'_\mathrm{dark}
    \!\! + \!\! (e_\mathrm{pol} \!\! + \!\! e_0 P_\mathrm{AP})(1 \!\!- \!\!G(\nu)) \right] e^{\nu}
    \!\! - \!\! e_0 p'_\mathrm{dark}
    }
    {Y_1^{L} \nu},
\end{split}
\end{equation}
where $Y_1^{L}$ is found in \cite[(G9)]{vasylyev2019satellite}, and the others in Lemma~1.

\section{Numerical Results \& Discussions}
\label{sec:results}

This section evaluates the performance of the proposed MRR-aided UAV-based QKD system over the FSO channels. 
{\color{black}
Parameters are listed in Table~\ref{tab:parameters_hybrid}. 
Monte Carlo simulations with $5 \times 10^6$ samples are used to validate the analytical results.
} 

\begin{table}[t]
\centering
\caption{\textsc{System Parameters\cite{Dabiri_TWC_2022, Dabiri_QKD_2025, Yudai}}}
\label{tab:parameters_hybrid}
\footnotesize
\setlength{\tabcolsep}{4pt}
\resizebox{.9\columnwidth}{!}{%
\begin{tabular}{lc|lc}
\hline
\textbf{Symbol} & \textbf{Value} & \textbf{Symbol} & \textbf{Value} \\ \hline
\multicolumn{4}{c}{\textbf{Ground Station (Bob)}} \\ \hline
$Z_{\mathrm{hg}}$ & 10 m & $\sigma_{\theta_e}$ & 200 $\mu$rad \\
$\lambda$ & 1550 nm & $\sigma_{\mathrm{AoA}}$ & 50 $\mu$rad \\
$w_0$ & 10 cm & $\mathcal{R}$ & $10^9$ pulses/s \\
$r_{\mathrm{g}}$ & 8 cm & $e_0$ & 0.5 \\
$\theta_{\mathrm{FoV}}$ & 100 $\mu$rad & $e_{\mathrm{pol}}$ & $0.033$ \\
$\eta$ & 0.6 & $p_{\mathrm{dark}}$ & $10^{-4}$ \\
$L_{\mathrm{g,NPBS_{R/T}}}$ & $0.1\mathrm{R}/0.9\mathrm{T}$ & $P_{\mathrm{AP}}$ & $0.02$ \\ \hline
\multicolumn{4}{c}{\textbf{UAV (Alice)}} \\ \hline
$\theta_{\mathrm{elv}}$ & {\color{black}$60^{\circ}$} & $L_{\mathrm{MQW}}$ & 20 dB \\
$A_{\mathrm{r}}$ & 1 cm$^2$ & $L_{\mathrm{u,NPBS_{R/T}}}$ & $0.5\mathrm{R}/0.5\mathrm{T}$ \\
$\sigma_\mathrm{\theta_o}$ & 2$^\circ$ & $L_{\mathrm{PBS}}$ & $0.5\mathrm{H}/0.5\mathrm{V}$ \\
$\xi_{\mathrm{div,2}}$ & 0.2 mrad & & \\ \hline
\multicolumn{4}{c}{\textbf{Quantum Channel}} \\ \hline
$Z$ & $1$ km & $V$ & 5 m/s \\
{\color{black}$\beta_{{\rm dB}, 1/2}$} & $0.43$ $\mathrm{dB/km}$ & $C_\text{n}^2(0)$ & $10^{-14}$ $\text{m}^{-2/3}$ \\
{\color{black}$\rho$} & {\color{black}0.7} & &  \\
\hline
\multicolumn{4}{c}{\textbf{BB84 Protocol}} \\ \hline
$\mu_\mathrm{n}$ & 0.3 & $\nu_\mathrm{n}$ & 0.09 \\
$s$ & 0.5 & $p$ & 0.75 \\
$f_\mathrm{e}$ & 1 & $d$ & 0.5 \\ \hline
\end{tabular}
}
\end{table}

\begin{figure}[t]
    \centering
    \includegraphics[width=\columnwidth]{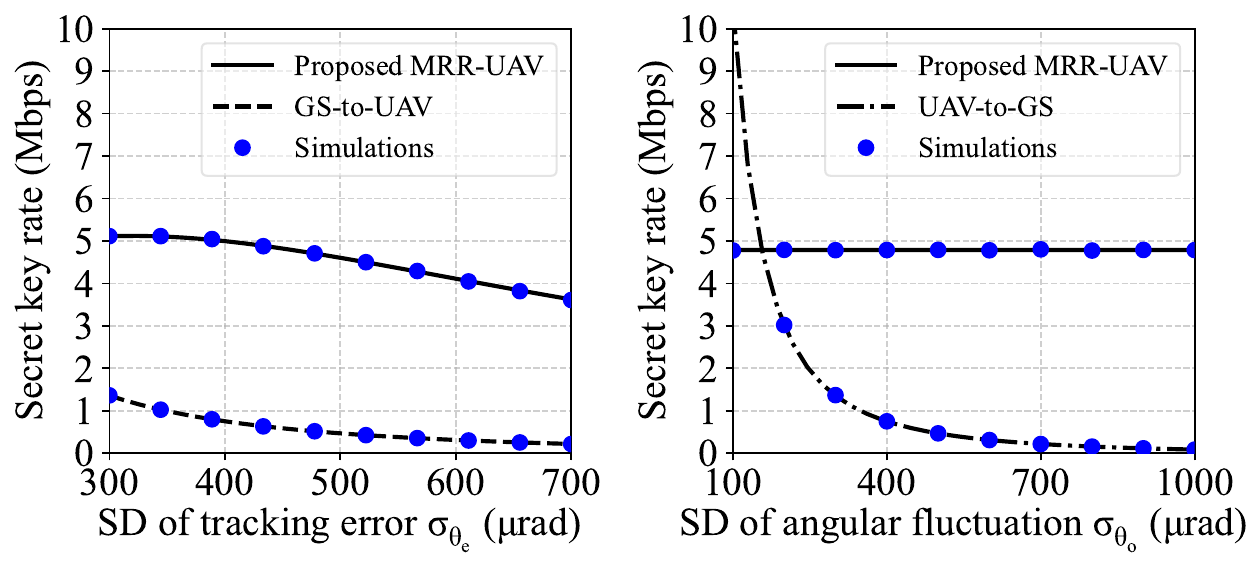}
    \vspace{-8mm}
    \caption{\color{black}SKR performance comparison under (a) GS tracking errors and (b) UAV hovering-induced angular fluctuations. The proposed MRR-UAV system outperforms the benchmark schemes in terms of robustness and SKR.}  
    \vspace{-6mm}
    \label{fig:cmp_w_benchmarks}
\end{figure}


First, we highlight the proposed system by comparing its SKR performance with that of conventional ones, including (i) GS-to-UAV and (ii) UAV-to-GS \cite{Dabiri_QKD_2025}, as depicted in Fig.~\ref{fig:cmp_w_benchmarks}. Particularly, Fig.~\ref{fig:cmp_w_benchmarks}(a) illustrates that the proposed system outperforms the GS-to-UAV system under severe GS tracking errors, since it can employ a broader beam waist to mitigate interrogating beam misalignment. In addition, Fig.~\ref{fig:cmp_w_benchmarks}(b) shows that the severe UAV hovering-induced angular fluctuations significantly degrade the SKR of the UAV-to-GS system, while the SKR of the proposed system remains unchanged. This is because the MRR can retroreflect the signal back along the incident path of the interrogating beam, thereby alleviating the impact of UAV angular fluctuations.
{\color{black}
Furthermore, the proposed system requires only simple additional hardware compared with the sophisticated ATP modules used in UAV-to-GS systems, which can weigh several kilograms \cite{Rabinovich2018free}. This can significantly reduce complexity, size, and weight on the UAV's side. The MQW-based MRR is expected to require only a few watts of power, compared to the tens of watts required by the ATP modules \cite{Rabinovich2018free}.
}
Finally, we verify the correctness of our theoretical framework, validated by simulations.

\begin{figure*}[!t]
    \centering
    \subfloat[]{
        \includegraphics[width=0.32\textwidth]{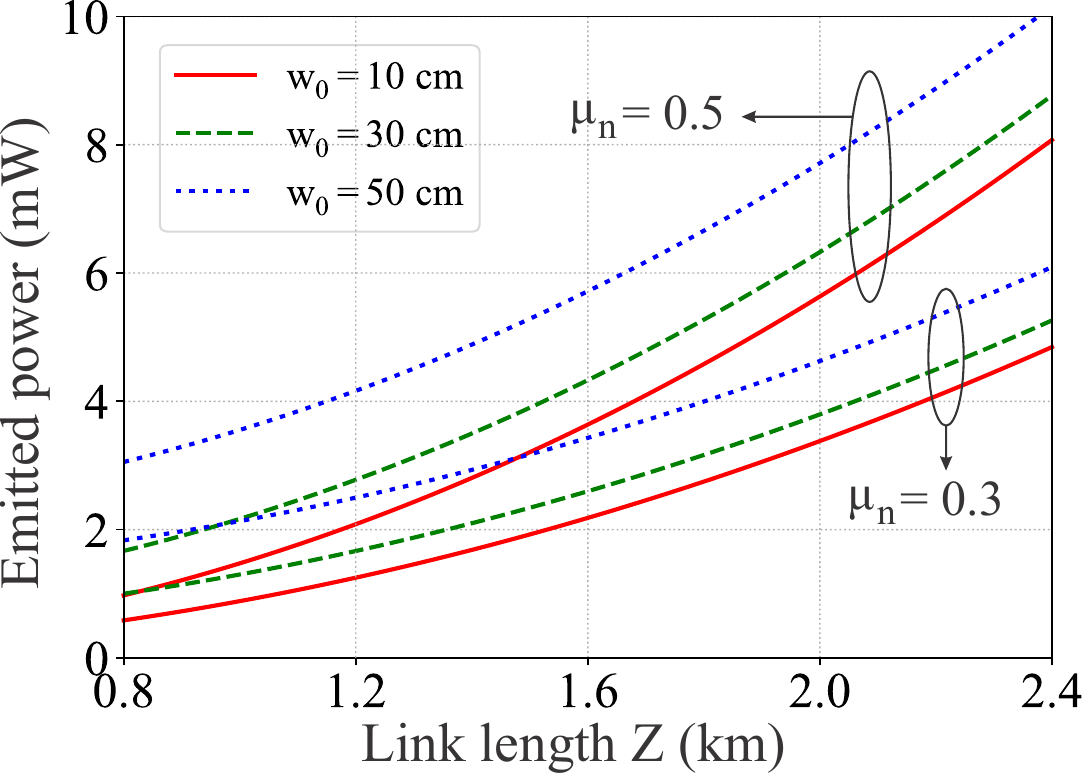}
        \label{fig:emit_power}
    }
    \hfill
    \subfloat[]{
        \includegraphics[width=0.3\textwidth]{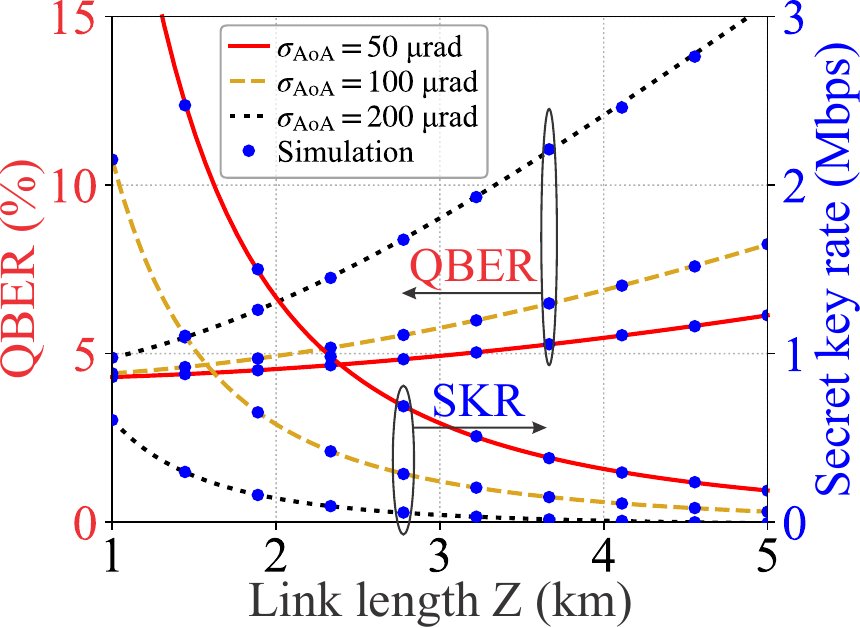}
        \label{fig:aoa_qber}
    }
    \hfill
    \subfloat[]{
        \includegraphics[width=0.32\textwidth]{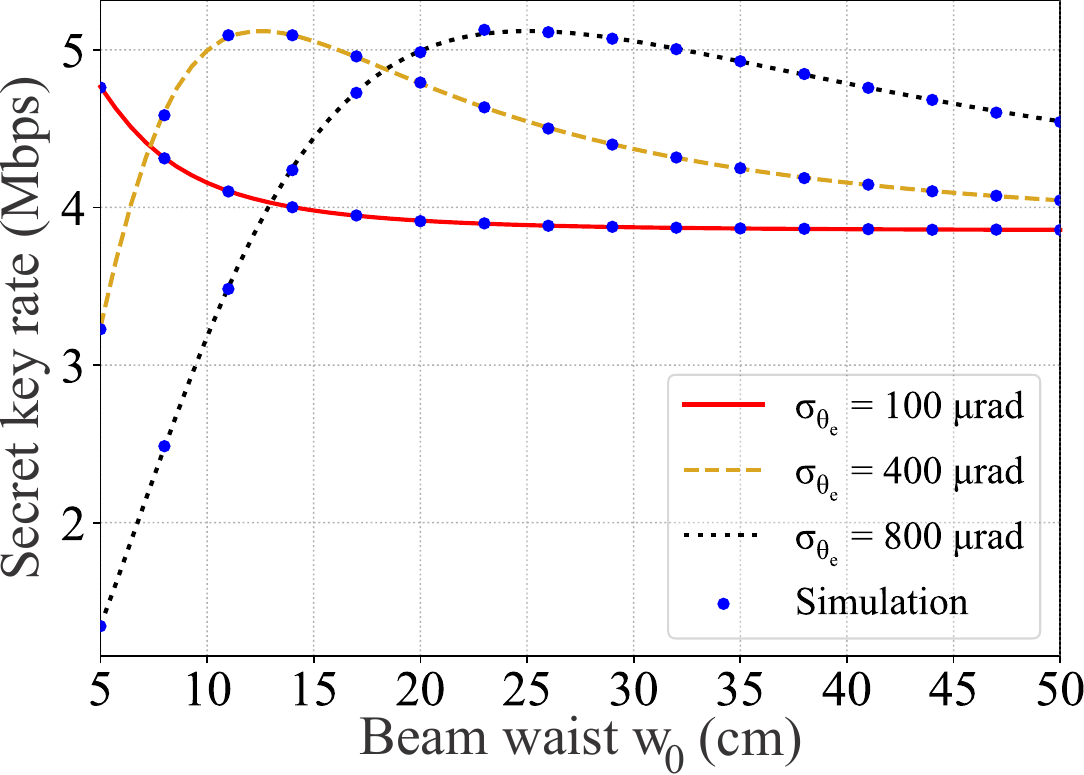}
        \label{fig:theta_skr}
    }
    \caption{\textcolor{black}{(a) Ground-emit power selection, (b) SKR and QBER for different link lengths, and (c) SKR versus beam waist for different pointing error conditions.}}
    \label{fig:trio_fig}
\end{figure*}







Next, we evaluate the performance of the proposed MRR-aided UAV-based FSO/QKD system in Fig.~\ref{fig:trio_fig}. Specifically, using Fig.~\ref{fig:trio_fig}(a), we can select the emission power level at the ground station to control the average number of photons reflected by the MRR over different values of distance $Z$ and beam waist $w_0$. For instance, when $Z$ = 2 km and $w_0$ = 30 cm, to maintain $\mu_\text{n}$ = 0.5 photons/pulse, the emission power should be adjusted as $P_\text{emit}$ = {\color{black}6} mW. Given $\mu_\text{n} = 0.3$ photons/pulse, we now analyze QBER and SKR performance with different link lengths and AoA deviation values in Fig.~\ref{fig:trio_fig}(b). As seen, increasing the link length or AoA deviation reduces the photon count rate. This increases the QBER, thereby lowering the SKR. From the results, we can determine a maximum link length to achieve a determined level of SKR for a given system's $\sigma_{\mathrm{AoA}}$. For instance, to achieve the SKR level of $0.3$ Mbps with a $\sigma_{\mathrm{AoA}}$ = $50~\mu$rad, the maximum operating link length is limited to {\color{black}$4$} km.
Finally, Fig.~\ref{fig:trio_fig}(c) demonstrates the SKR versus beam waist with different pointing error conditions. Using this figure, we can determine the optimal beam waist value to maximize the SKR. For example, the optimal beam waist values are $5$, $15$, and {\color{black}$25$} cm for SD of tracking error values $\sigma_{\mathrm{\theta_e}} = 100, 400, 800$ $\mu$rad, respectively.

\section{Conclusions}
\label{sec:conclusion}
This paper investigated the performance of the MRR-aided UAV-based FSO/QKD system over realistic free-space quantum channels. Analytical expressions for channel PDT, QBER, and SKR were derived. The obtained results demonstrated the feasibility of integrating MRR into UAV-based FSO/QKD systems and confirmed its {\color{black} effectiveness over the conventional UAV-ground} approaches. Furthermore, the results provided useful design insights for selecting system parameters toward practical QKD deployment. {\color{black}Future work will focus on hardware-oriented real-time validation, the effect of correlated channel impairments (e.g., how the UAV’s angular fluctuation effects MRR's photon loss and AoA), and advanced attack models, such as phase randomization and Trojan-horse attacks.}

\appendices
\section{Proof of Lemma 1}
\label{Appenx_A}

The effective channel transmittance can be written as 
\begin{align}
    \tau 
    &= \eta \overbrace{\frac{I_{\text{a},1}}{\mathbb{E}[I_{\text{a},1}]} I_{\text{a},2} \frac{\eta_{\theta}}{\mathbb{E}[\eta_{\theta}]}}^{\tau_{\mathrm{LN}}} \frac{\tau_{\text{p},1}}{\mathbb{E}[\tau_{\text{p},1}]} \tau_{\text{p},2} \tau_{\text{l},2} {\color{black}L_{\mathrm{g,NPBS_{T}}}}  \tau_\text{FoV} = \frac{\tau' \tau_\text{FoV}}{\mathbb{E}[\tau_{\text{p},1}]}.
\end{align}
As $\tau_{\mathrm{LN}}$ are the product of three log-normal random variables, we have {\color{black}$\tau_{\mathrm{LN}} \sim LN(- \frac{\sigma^2_\mathrm{R,1}}{2} - \frac{\sigma^2_\mathrm{R,2}}{2} - \frac{1}{2 s_\mathrm{m}^2}, \sigma^2_\mathrm{R,1} + \sigma^2_\mathrm{R,2} + 2 \rho \sigma_\mathrm{R,1} \sigma_\mathrm{R,2} + s_\mathrm{m}^2)$}.
Thus, we can derive the closed-form expression of $\tau'$ by following the steps in~\cite[Appendix~A]{Dabiri_TWC_2022}. Finally, by incorporating $\tau'$ with the Bernoulli distribution of $\tau_{\mathrm{FoV}}$, we obtain the PDT of $\tau$ as in~\eqref{eqn:eff_channel}.

\section{Proof of Proposition 1}
\label{Appenx_B}

After a few steps of transformation, QBER is given as 
\begin{align}
    \label{eqn:qber_app}
    \langle E_{\mu_\text{n}} \rangle \!\! = \!\! \frac{e_\mathrm0 p'_\mathrm{dark}\mathrm \!\! + \!\! \big(e_\mathrm{pol} \!\! + \!\! e_\mathrm0 P_\mathrm{AP}\big)\big[1 \!\! - \!\! \int_0^\infty \!\!  e^{-\mu_\text{n} \tau} \, \!\! f_\tau(\tau) \mathrm{d} \tau\big]}{p'_\mathrm{dark} + (1 + P_\mathrm{AP})\big[1 - \int_0^\infty  e^{-\mu_\text{n} \tau} \, f_\tau(\tau) \mathrm{d} \tau\big]}.
\end{align}
Let $G \left( \mu_\mathrm{n} \right)$ be the integral in \eqref{eqn:qber_app}. To derive its approximation, we first substitute \eqref{eqn:eff_channel} and apply the change of variable $x = \frac{\ln \tau + C_5}{\sqrt{C_1}}$. Then, by applying the lower incomplete gamma function \cite[(3.9)]{gamma_func_new} and the Gauss-Hermite quadrature approximation, we derive $G \left( z \right)$ as in~\eqref{G_approx}.
\balance

\bibliographystyle{IEEEtran}
\bibliography{references}

@article{UAV_survey_2026,
  title={A Survey on \text{DRL}-Based \text{UAV} Communications and Networking: \text{DRL} Fundamentals, Applications and Implementations},
  author={Zhao, Wei and others},
  journal={IEEE Commun. Surveys Tut.},
  year={2025\color{black}},
  publisher={IEEE}
}

@misc{Dabiri_QKD_2025,
      title={A Unified Framework for \text{UAV}-Based Free-Space Quantum Links: Beam Shaping and Adaptive Field-of-View Control}, 
      author={Mohammad Taghi Dabiri and Mazen Hasna and Saif Al-Kuwari and Khalid Qaraqe},
      year={2025\color{black}},
      eprint={2506.20336},
      archivePrefix={arXiv},
      primaryClass={eess.SP},
      url={https://arxiv.org/abs/2506.20336}, 
}

@ARTICLE{Dabiri_TWC_2022,
  author={Dabiri, Mohammad Taghi and others},
  journal={IEEE Trans. Wireless Commun.}, 
  title={Modulating Retroreflector Based Free Space Optical Link for \text{UAV}-to-Ground Communications}, 
  year={Oct. 2022\color{black}},
  volume={21},
  number={10},
  pages={8631-8645},
}

@article{Vallone2015experimental,
  title={Experimental Satellite Quantum Communications},
  author={Vallone, Giuseppe and others},
  journal={Phys. Rev. Lett.},
  volume={115},
  number={4},
  pages={040502},
  year={Jul. 2015},
}

@article{Rabinovich2018free,
  title={Free space quantum key distribution using modulating retro-reflectors},
  author={Rabinovich, William S and others},
  journal={Optica Opt. Express},
  volume={26},
  number={9},
  pages={11331--11351},
  year={Apr. 2018},
}

@article{ref_8,
  title={Free-Space \text{QKD} with Modulating Retroreflectors Based on the \text{B92} Protocol},
  author={Minghao Zhu and others},
  journal={Entropy},
  year={Jan. 2022\color{black}},
  volume={24},
}

@ARTICLE{Liu_MRR_AO_2017,
  author={\color{black}Changan Liu and others},
  journal={Optica Appl. Opt.}, 
  title={Method to uniformly diverge a reflected beam with a plano-concave lens located in front of a cube-corner retroreflector}, 
  year={Feb. 2017},
  volume={56},
  number={5},
  pages={1333-1338},
}

@INPROCEEDINGS{Yudai,
  author={Takihara, Yudai and Le, Hoang D. and Nguyen, Cuong T. and Pham, Anh T.},
  booktitle={IEEE Int. Conf. Advanced Technol. Commun.}, 
  title={Qiskit-Enabled Simulation of Quantum Key Distribution over Starlink Satellite Networks}, 
  year={2025},
  volume={},
  number={},
  pages={1-6\color{black}},
}

@article{trinh2025optical,
  title={{Optical RISs improve the secret key rate of free-space QKD in HAP-to-UAV scenarios}},
  author={Trinh, Phuc V and Sugiura, Shinya and Xu, Chao and Hanzo, Lajos},
  journal={IEEE J. Sel. Areas Commun.},
  volume={43},
  number={8},
  pages={2747-2764},
  year={Aug. 2025\color{black}},
  publisher={IEEE}
}

@article{vasylyev2019satellite,
  title     = {Satellite-mediated quantum atmospheric links},
  author    = {\color{black}Vasylyev, Dmytro and Vogel, W and Moll, Florian},
  journal   = {Phys. Rev. A},
  volume    = {99},
  number    = {5},
  pages     = {053830},
  year      = {May 2019\color{black}},
  publisher = {APS}
}

@article{trinh2021experimental,
  title={Experimental channel statistics of drone-to-ground retro-reflected \text{FSO} links with fine-tracking systems},
  author={Trinh, Phuc V. and others},
  journal={IEEE Access},
  volume={9},
  pages={137148--137164},
  year={Oct. 2021},
  publisher={IEEE}
}

@article{dabiri2025novel,
  title={A novel {MRR-UAV} based relay with optical network coding: A comparative study with optical {IRS} and conventional {UAV} relaying},
  author={Dabiri, Mohammad Taghi and Hasna, Mazen},
  journal={IEEE J. Sel. Areas Commun.},
  volume={43},
  number={5},
  pages={1607--1620},
  year={May 2025\color{black}},
  publisher={IEEE}
}

@article{ismail2025finite,
  title   = {Finite-Size and Modulation Optimization in {CV-QKD} over {UAV}-Based {FSO} Link with Adaptive Optics},
  author  = {Ismail, Tawfik and Sabeeh, Ala H. and Yasser, Mahmoud and Alshaer, Nancy},
  journal = {IEEE Commun. Lett.},
  year    = {Dec. 2025\color{black}},
  volume  = {29},
  number  = {12},
  pages   = {1--5},
  publisher = {IEEE}
}

@ARTICLE{why_mqw_2021,
  author={Crisanto Quintana and others},
  journal={IEEE/Optica J. Lightw. Technol.}, 
  title={A High Speed Retro-Reflective Free Space Optics Links With \text{UAV}}, 
  year={2021\color{black}},
  volume={39},
  number={18},
  pages={5699-5705}
}

@article{ata2022performance,
  title={Performance of integrated ground-air-space {FSO} links over various turbulent environments},
  author={Ata, Yal{\c{c}}{\i}n and Alouini, Mohamed-Slim},
  journal={IEEE Photonics J.},
  volume={14},
  number={6},
  pages={1--16},
  year={Dec. 2022 \color{black}},
  publisher={IEEE}
}

@Inbook{gamma_func_new,
author="Fraczek, Markus Szymon",
title="The Gamma Function and the Incomplete Gamma Functions",
bookTitle="Selberg Zeta Functions and Transfer Operators: An Experimental Approach to Singular Perturbations",
year="2017",
publisher="Springer International Publishing",
address="Cham",
pages="39--42",
isbn="978-3-319-51296-9",
doi="10.1007/978-3-319-51296-9_3",
url="https://doi.org/10.1007/978-3-319-51296-9_3"
}

\end{document}